%% file: main.tex
\documentclass[11pt]{article}

\usepackage[T1]{fontenc}
\usepackage[utf8]{inputenc}
\usepackage[a4paper,margin=1in]{geometry}
\usepackage{xcolor}
\usepackage{lipsum}

\usepackage{amsmath,amssymb,amsthm,mathtools,bbm}
\usepackage{microtype}
\usepackage{xspace}
\usepackage{enumitem}
\usepackage{booktabs}
\usepackage{multirow}

\usepackage{tikz}
\usepackage{wrapfig}
\usepackage{comment}

\usepackage[ruled,vlined,linesnumbered]{algorithm2e}

\usepackage{framed}
\usepackage[framemethod=tikz]{mdframed}

\usepackage{hyperref}

\usepackage{aliascnt}

\newtheorem{theorem}{Theorem}[section]
\newtheorem{question}{Question}

\newaliascnt{lemma}{theorem}
\newtheorem{lemma}[lemma]{Lemma}
\aliascntresetthe{lemma}

\newaliascnt{claim}{theorem}
\newtheorem{claim}[claim]{Claim}
\aliascntresetthe{claim}

\newaliascnt{corollary}{theorem}
\newtheorem{corollary}[corollary]{Corollary}
\aliascntresetthe{corollary}

\newaliascnt{proposition}{theorem}

\aliascntresetthe{proposition}

\newaliascnt{observation}{theorem}
\newtheorem{observation}[observation]{Observation}
\aliascntresetthe{observation}

\theoremstyle{definition}
\newaliascnt{definition}{theorem}
\newtheorem{definition}[definition]{Definition}
\aliascntresetthe{definition}

\theoremstyle{remark}
\newaliascnt{remark}{theorem}
\newtheorem{remark}[remark]{Remark}
\aliascntresetthe{remark}

\usepackage[nameinlink,noabbrev,capitalise]{cleveref}

\crefname{theorem}{Theorem}{Theorems}
\Crefname{theorem}{Theorem}{Theorems}

\crefname{lemma}{Lemma}{Lemmas}
\Crefname{lemma}{Lemma}{Lemmas}

\crefname{claim}{Claim}{Claims}
\Crefname{claim}{Claim}{Claims}

\crefname{corollary}{Corollary}{Corollaries}
\Crefname{corollary}{Corollary}{Corollaries}

\crefname{proposition}{Proposition}{Propositions}
\Crefname{proposition}{Proposition}{Propositions}

\crefname{observation}{Observation}{Observations}
\Crefname{observation}{Observation}{Observations}

\crefname{definition}{Definition}{Definitions}
\Crefname{definition}{Definition}{Definitions}

\crefname{remark}{Remark}{Remarks}
\Crefname{remark}{Remark}{Remarks}

\crefname{algocf}{Algorithm}{Algorithms}
\Crefname{algocf}{Algorithm}{Algorithms}

\newcommand{\V}{V}
\newcommand{\E}{E}
\newcommand{\G}{G}
\newcommand{\Hgraph}{H}

\newcommand{\dist}{\mathrm{dist}}
\newcommand{\distG}{\dist_{\G}}
\newcommand{\distH}{\dist_{\Hgraph}}

\newcommand{\Pxy}{P_{x,y}}

\title{Greedy Completion for Weighted $(\alpha,\beta)$-Spanners}
\author{Elad Tzalik\thanks{Weizmann Institute. Email: \texttt{elad.tzalik@weizmann.ac.il}. Supported by the Adams Fellowship Program of the Israel Academy of Sciences and Humanities.}}
\date{}

\begin{document}
\maketitle

\begin{abstract}
    We study $(\alpha,\beta)$-spanners for weighted graphs. We propose a simple greedy completion procedure which starts from a sparse initial graph, and repeatedly fixes pairs of vertices with a bad stretch, generalizing Knudsen's additive completion [SWAT '14]. As an application, we
construct $(k,k-1)$-spanners for weighted graphs of size $\tilde{O}(n^{1+1/k})$, which were previously unknown.
\end{abstract}


\input{intro}

\input{completion}

\input{analysis}

\input{Conclusions}

\bibliographystyle{alpha}
\bibliography{references}

\appendix

\input{Even}

\end{document}

%% file: intro.tex
\section{Introduction}

Let $G$ be an $n$-vertex weighted graph.
A \emph{spanner} of $G$ is a subgraph that approximately preserves distances.
Formally, a subgraph $H$ is a $t$-spanner of $G$ if $\dist_H (u,v) \leq t \cdot \dist_G (u,v)$ for all $u,v\in V$ (where $\dist_X$ denotes the shortest path distance in a graph $X$).
Since their introduction by Peleg and Ullman \cite{PelegU88} and Peleg and Sch\"{a}ffer \cite{PelegS89}, spanners have been found to be extremely useful for a wide variety of applications, including network tasks like routing and synchronization \cite{PelegU89a,ThorupZ01,Cowen01,CowenW04,PelegU88}, preconditioning linear systems \cite{ElkinEST08}, distance estimation \cite{ThorupZ05,BaswanaS04}, and many others.   

The first tight construction, and one of the most influential spanner constructions, was given by \cite{AlthoferDDJS93}, which described a simple greedy algorithm that produces a $(2k-1)$-spanner of a graph $G$, of size $O(n^{1+\frac{1}{k}})$ edges, which is best possible\footnote{The $O(n^{1+\frac{1}{k}})$ upper bound is tight \emph{assuming the girth conjecture of Erd\H{o}s}. Nevertheless, the algorithm of \cite{AlthoferDDJS93} is optimal \emph{unconditionally}.}.

To bypass this size-stretch tradeoff, much work has gone into finding better ways to parameterize how shortest paths in $G$ may stretch in $H$. The most well-studied such generalization is $(\alpha,\beta)$-spanners which allow a multiplicative stretch $\alpha$ (which may be viewed as the effective stretch for pairs that are sufficiently far apart in the graph), plus an additive stretch $\beta$ which is used to mask the local barrier of the stretch of the endpoints of an edge in a high girth graph. Let $W_{max}(G)$ denote the maximum weight of an edge in $G$.

\begin{definition}[$(\alpha,\beta)$-spanner]\label{def:alpha-beta-spanner}
    A subgraph $H \subseteq G$ is an \emph{$(\alpha,\beta)$-spanner} of $G$ if for every $u,v \in V$,
    \[
    \dist_H(u,v) \leq \alpha \cdot \dist_G(u,v) + \beta \cdot W_{\max}(G).
    \]
\end{definition}

\paragraph*{The quest to handle weights.}
Spanners are often used to compress metric spaces that correspond to weighted
input graphs, see e.g. \cite{CaiK97,DobsonB14,MarbleB13,SalzmanSAH14}. While standard multiplicative spanners can handle weights easily, algorithms that construct $(\alpha,\beta)$-spanners for unweighted graphs typically take non-trivial effort to be adapted to handle weights. A notable example is the $+6$ spanner \cite{BaswanaKMP10,knudsen2014additive}, where extending the construction to weighted graphs required three separate works \cite{AhmedBSKS20,ElkinGN23,LaL24}. 
Initially, the work of Cohen \cite{Cohen00} constructed weighted $(1+\varepsilon,\beta)$ spanners, which was then improved by \cite{Elkin01}. 
In recent years, there has been a renewed and growing interest in understanding spanners for weighted graphs. 
The work of \cite{ElkinGN22} achieved constructions of weighted $(\alpha,\beta)$ spanners with \emph{local} stretch, meaning that $W_{max}(G)$ is replaced by $W_{max}(P)$, which stands for the smallest maximum edge weight along a shortest path $P$ between nodes. 
A sequence of recent works \cite{AhmedBSKS20,ElkinGN23,AhmedBHKS21,LaL24} gave tight bounds for weighted \emph{additive} spanners. A key ingredient in these works is the work of Knudsen \cite{knudsen2014additive}, which describes a simple algorithm to obtain additive spanners greedily.

\paragraph*{Handling weights: beyond additive stretch.} The main downside of additive spanners is that they may be large. The seminal work of Abboud and Bodwin \cite{DBLP:journals/jacm/AbboudB17} proved that such spanners may require $\Omega(n^{\frac{4}{3}-\varepsilon})$ edges for any $\epsilon>0$. 
This motivates the study of weighted $(\alpha,\beta)$-spanners for $\alpha>1$. Unfortunately, this case is far less understood with only \emph{one prior work}\footnote{Putting aside the case of $\beta=0$ for which the classical multiplicative spanners of \cite{AlthoferDDJS93} handle weights.} providing such constructions, by Elkin, Gitlitz and Neiman \cite{ElkinGN22}, constructing a weighted $\left(O(1),k^{O(1)} \right)$- spanner with $O(n^{1+\frac{1}{k}})$ edges, as well as weighted $\left(1+\varepsilon,\beta \right)$ spanners with $\beta=O\left(\frac{\log(k)}{\varepsilon}\right)^{\log (k)}$. 

\paragraph*{The "local" recipe for $(\alpha,\beta)$ spanners, and how it fails to handle weights.}
We now explain why existing constructions of $(\alpha,\beta)$-spanners fail to handle weighted graphs, e.g. those of \cite{BaswanaKMP10,Ben-LevyP20,ChechikL26,PopovaTzalikITCS}. When the graph $G$ is unweighted, obtaining an $(\alpha,\beta)$ spanner can be done "locally'' by guaranteeing good stretch for paths of fixed length. For example, to obtain a $(k,k-1)$-spanner for an \emph{unweighted} graph $G$, one can take a standard spanner $H_1$ that guarantees a stretch of $2k-1$ for a single edge, and a subgraph $H_2$ which guarantees a stretch of $k$ for vertices of distance $2$. $H_1\cup H_2$ is the desired spanner - indeed given $x,y$ with a shortest path $P:=P_{x,y}$ with $t$ edges in $G$, break $P$ into $\lfloor t/2\rfloor$ paths of length $2$ with an extra edge remaining if $t$ is odd, and use $H_2$ to replace the length $2$ paths, and $H_1$ to replace the remaining single edge. 
As it turns out, such a ''recipe" cannot yield a $(k,k-1)$-spanner for weighted graphs. In a recent work, Popova and Tzalik \cite{PopovaTzalikITCS} proved that in any weighted graph $G$ the best approximation one can provide for a path $P$ of length $2$ is $W_{min}(P)+(2k-1)W_{max}(P)$\footnote{$W_{max}(P)$ denotes the maximum edge weight along $P$, $W_{min}$ the minimum one.}, assuming the desired size bound on the spanner is $O(n^{1+1/k})$, hence in particular, such an approach may incur an additive $+(2k-2)W_{max}(G)$ error on every second edge along the path, whereas a $(k,k-1)$-spanner allows a $(k-1)W_{max}(G)$ additive error only once along the path. 

Putting it all together, recent works have investigated how to construct weighted $(\alpha,\beta)$-spanners, and while the important case of additive spanners ($\alpha=1$) has been resolved in \cite{AhmedBSKS20,ElkinGN23,AhmedBHKS21,LaL24}, the \emph{general} case of $\alpha>1$ has resisted progress, with only a single known construction by \cite{ElkinGN22}. 
Due to the lack of a desired "canonical'' construction of $(\alpha,\beta)$-spanners, in a survey on graph spanners from 2020, the following question (open question 4.7.2 \cite{AhmedBSHJKS20}) was asked:

\begin{question}
     Is there a greedy algorithm to produce additive spanners, $(\alpha,\beta)$-spanners, subsetwise spanners, or more generally pairwise spanners?
\end{question}
In this paper we answer this question in the affirmative for $(\alpha,\beta)$-spanners by introducing a
greedy completion algorithm that extends Knudsen’s approach to $(\alpha,\beta)$-spanners.

Using the new completion method, we obtain as an application a $(k,k-1)$-spanner for any graph $G$, with essentially tight size $\Tilde{O}(n^{1+1/k})$\footnote{The size bounds are tight assuming the girth conjecture, as a graph with girth $>2k$ has only itself as a $(k,k-1)$ spanner, even in the unweighted setup. }. To our knowledge, such spanners were not previously known.

\begin{theorem}\label{thm:main}
    For $k\geq 2$ and any weighted graph $G$, there exists a $(k,k-1)$ spanner $H$ of $G$ of size $\Tilde{O}(n^{1+1/k})$. Moreover, $H$ can be computed in polynomial time.
\end{theorem}

The $(k,k-1)$-spanners we construct can be viewed as a spanner whose multiplicative stretch is roughly half the multiplicative stretch of the classical greedy spanner, while still allowing for a single bottleneck, reflected in the additive $(k-1)$ term. To our knowledge, prior to this work there was no known construction of an $(\alpha,\beta)$ spanner that improves upon the stretch of the greedy algorithm for weighted graphs, and with similar size.\footnote{The work of \cite{PopovaTzalikITCS} does provide spanners which improve upon the multiplicative $(2k-1)$-spanner of \cite{AlthoferDDJS93} for non-adjacent pairs of vertices. The key distinction is that in this work, the stretch is roughly $\frac{2k-1}{k}$, aka \emph{twice} better then a multiplicative spanner, yet their methods cannot give such effective bounds. }

\paragraph*{Conceptual contribution.}
Beyond the $(k,k-1)$ construction, the main contribution of this work is a new
\emph{global greedy completion paradigm} for constructing weighted $(\alpha,\beta)$-spanners. 
The algorithm extends Knudsen’s additive completion by introducing a notion of
\emph{minimal segmentations} of shortest paths, and repairing only the bottlenecks that
violate the target stretch. This avoids the over-correction inherent in naive greedy
approaches, and crucially enables handling weighted graphs, where local path-based
constructions fail. Interestingly, the analysis gives insight towards the actual structure of the replacement shortest paths, beyond just the $(\alpha,\beta)$ stretch guarantee, and leads to interesting new direction, described in \Cref{sec:final-remarks}.

\paragraph*{Comparison to~\cite{PopovaTzalikITCS}.}
The two works study different aspects of greedy spanner constructions.
The focus of~\cite{PopovaTzalikITCS} is on greedy constructions of $d \to r$ spanners in
unweighted graphs, namely subgraphs in which every pair of vertices at
distance exactly $d$ in $G$ has distance at most $r$ in the spanner.
Such constructions naturally serve as \emph{local} ingredients in the
design of unweighted $(\alpha,\beta)$-spanners, and their analysis relies
on clustering, ball-growth, and path-buying arguments around short paths.
In contrast, the present work studies weighted $(\alpha,\beta)$-spanners
directly. As explained previously, the more challenging weighted setting, is inherently \emph{global} and requires new methods. Consequently, both the algorithm and the analysis in the present paper are completely different. 


\section{Technical Overview}

Throughout the technical overview we assume for simplicity the graph $G$ is unweighted. The fact that the proposed algorithm is global, in the sense that it is independent of the actual number of edges along a path, makes it possible to handle weighted graphs.

\subsection{Knudsen's algorithm}

Knudsen's algorithm takes as input a graph $G$, and an integer $\beta$, and produces a $+\beta$ spanner of $G$. It consists of $2$ main phases:
\begin{itemize}
    \item(Initialize) Compute an initial graph $H_0$, obtained by going over all $v\in V$ and adding $n^{\gamma}$ arbitrary neighboring edges to $v$. Set $H\gets H_0$.

    \item(Additive Completion) While there exist $x,y$ such that $\dist_H(x,y) > \dist_G(x,y)+\beta $: Find an $x$-to-$y$ shortest path $P_{x,y}$. Add to $H$ all edges of $P_{x,y}$ missing from $H$.
\end{itemize}

We now give a high-level overview of the analysis. The analysis hinges on showing that in each completion step, if in the $i^{th}$ step $t_i$ new edges along $P_{x,y}$ are added, then $\Omega(t_i n^{1-\gamma})$ pairs\footnote{This doesn't hold for all $\gamma>0$, but for $\beta=2,6$ one may set $\gamma=1/2,1/3$ respectively, and get this way a spanner of optimal size.} of vertices have their distance substantially improved by the completion step, meaning that if $H$ is the graph before the completion step, and $H'$ is the graph after the completion step, then for $\Omega(t_i n^{1-\gamma})$ pairs $(u,v)$ we have $\dist_H(u,v)-\dist_{H'}(u,v)\geq 1$. 
Moreover, in the first $H$ for which $(u,v)$ improves, we have $\dist_{H'}(u,v)\leq \dist_G(u,v)+O(1)$. Combining the fact that a pair $(u,v)$ is within $O(1)$  of its distance in $G$, after the first improvement, there can be $O(1)$ additional steps for which $\dist_H(u,v)-\dist_{H'}(u,v)\geq 1$. Therefore one gets that $\sum_i \Omega(t_i n^{1-\gamma}) = O(n^{2})$, hence $\sum_i t_i = O(n^{1+\gamma})$.
In conclusion the final size bound of the output spanner $H$ is:
\[|E(H)| = O(n^{1+\gamma}) + \sum_i t_i = O(n^{1+\gamma}).\]

\subsection{Defining the $(\alpha,\beta)$-completion}

As usual with greedy algorithms, correctness is straightforward; the main challenge is to show that the final output is sparse.  We now describe the main conceptual contribution of the paper, which lies in finding the right way to define the $(\alpha,\beta)$-completion in a way that preserves correctness and, crucially, allows one to prove a good bound on the size of the output graph. To do so, we incorporate two delicate adaptations to Knudsen's algorithm. We describe the algorithm progressively, through a sequence of increasingly refined attempts. Assume we are given a graph $G$, parameters $(\alpha,\beta)$, and some initial graph $H_0$ to be completed to an $(\alpha,\beta)$-spanner of $G$. A natural first attempt is to apply Knudsen's original idea directly.

\paragraph*{Attempt 1 (unsuccessful): naive $(\alpha,\beta)$-completion} 

While there exist $x,y$ such that: \[\dist_H(x,y) > \alpha \cdot \dist_G(x,y)+\beta.\] 
Find an $x$-to-$y$ shortest path $P_{x,y}$. Add to $H$ all edges of $P_{x,y}$ missing from $H$.

The naive $(\alpha,\beta)$-completion successfully fixes pairs $x,y$ with bad stretch --- but at what cost? The key problem with the naive completion is its \emph{over-correction}. If initially $x,y$ satisfied $\dist_H(x,y)=\alpha \cdot \dist_G(x,y)+\beta+1$, it seems plausible that a small fix could have done the job. On the other hand, by adding all missing edges along $P_{x,y}$ we make the distance between $x,y$ in the new graph equal to their actual distance in $G$ - this indeed fixes the bad stretch of $x,y$, but it seems to take it too far.

\paragraph*{Attempt 2 (unsuccessful): $(\alpha,\beta)$-completion by fixing the bottlenecks} 

The key observation is that we can find \emph{a mild} fix to the distance of $x,y$ by only fixing the bottlenecks along a shortest path. Let $P_{x,y}$ be an $x$-to-$y$ shortest path, and let $e$ be an edge on $P_{x,y}$. Say $e=\{u,v\}$ has $\alpha$-distant endpoints in $H$ if $\dist_H(u,v) > \alpha \cdot \dist_G(u,v)$. The key observation is that if  \emph{all $\alpha$-distant} edges are added along the path, then in the new graph $H'$ $\dist_{H'}(x,y)\leq \alpha\cdot \dist_G(x,y)$, hence this still fixes the pair $x,y$, but now only the ``bottlenecks'' along the path are added to $H$, which helps ensure that only a few edges are added.

Formally the new $(\alpha,\beta)$-completion is:

\noindent While there exist $x,y$ such that: \[\dist_H(x,y) > \alpha \cdot \dist_G(x,y)+\beta.\] 
Find an $x$-to-$y$ shortest path $P_{x,y}$. Add to $H$ all the $\alpha$-distant edges of $P_{x,y}$ missing from $H$.

This idea is close to the final algorithm, but one additional ingredient is needed to control the total number of added edges.

\paragraph*{Attempt 3 (successful): incorporating shortcuts} 

The last optimization is to allow bypassing bottlenecks along $P_{x,y}$ if possible. The key observation is that it may be the case that for example $e_1=(u,v),e_2=(v,w),e_3=(w,x)$ are three consecutive edges along $P_{x,y}$ which are all $\alpha$-distant, yet $\dist_H(u,x)\leq \alpha \cdot \dist_G(u,x)$. Hence if we can find in $H$ a path of stretch $\alpha$ between endpoints of $\alpha$-distant edges which we can use to "jump over'' the bottlenecks - we do so. 

In other words, we partition $P_{x,y}$ into segments that are either a single $\alpha$-distant edge, or pairs of vertices with stretch $\leq \alpha$ in $H$,
and choose one whose segments cannot be merged further; we call such segmentations \emph{minimal}. Finally the $(\alpha,\beta)$-completion is:

\noindent While there exist $x,y$ such that: \[\dist_H(x,y) > \alpha \cdot \dist_G(x,y)+\beta.\] 
Find an $x$-to-$y$ shortest path $P_{x,y}$. Compute a minimal segmentation of $P_{x,y}$ w.r.t $H$. Add to $H$ all the $\alpha$-distant edges in the chosen segmentation.

\subsection{The Baswana-Sen algorithm, and its properties}
The Baswana-Sen algorithm~\cite{baswana2007simple} is a useful tool for computing \emph{multiplicative} spanners, by clustering vertices to centers. The set of centers is defined inductively by $S_0=V$, and $S_i=\mathbf{Sample}(S_{i-1},p)$ where $p=\tilde{O}\left(\frac{1}{n^{1/k}}\right)$. For each vertex $u \in V$ and level $i$, we denote by $s_u^i \in S_i$ the center of $u$ at level $i$. We do not review the algorithm here, and refer the reader to~\cite{baswana2007simple}; for a local description see~\cite{parter2025local}. We now record a key invariant preserved by the Baswana-Sen algorithm.

\begin{theorem}[\cite{baswana2007simple}]\label{thm:baswana-sen}
    For any graph $G$, there exists a subgraph $H' \subseteq G$, with $\tilde{O}(n^{1+1/k})$ edges such that for every $i \in \{1,\ldots,k\}$, and every $\{u,v\}=e \in E(G)$ one of the following holds:
    \begin{enumerate}
        \item $\dist_{H}(u,v)\leq (2i-1)w(e)$; or

        \item $\exists s_u^{i}, s_v^{i} \in S_{i}$ such that $\dist(u,s_u^{i}) \leq i \cdot w(e)$ and $\dist(v,s_v^{i}) \leq i \cdot w(e)$.
    \end{enumerate}
\end{theorem}

We also record the following observation which follows from standard concentration bounds and a union bound over $k$:

\begin{observation}
    With high probability for all $i \in \{0,\ldots k-1\}$, we have $|S_i|=\tilde{O}(n^{1-i/k})$.
\end{observation}

\subsection{The $(k,k-1)$-spanner algorithm, and its analysis}

To construct the $(k,k-1)$-spanner we take the $(k,k-1)$-completion of the initial graph $H_0$ defined as follows: Run the Baswana-Sen algorithm \cite{baswana2007simple} that produces a $(2k-1)$-spanner $H$, and set $H_0\gets H$ as the initial graph for the greedy completion. Using clustering as an initialization step is natural, and common in spanner constructions, see e.g. \cite{knudsen2014additive,AhmedBSKS20,ElkinGN23}.

The analysis follows the high-level structure of Knudsen's original analysis, but requires several new adaptations. For simplicity we assume $k$ is odd, and set $R:=\frac{k-1}{2}$. Let $t_i$ be the number of edges added in the $i^{th}$ completion step. The analysis tracks distances only for pairs of \emph{centers}, denoted by $(s,s')$, with $s,s' \in S_{R}$. We now specify how we track the progress of the pair $(s,s')$:
\begin{itemize}
    \item(Set-off) The first iteration for which $s,s'$ are centers of vertices along $P_{x,y}$ where $x,y$ are the vertices chosen in the completion step. In \Cref{clm:post-setoff}, we show that in subsequent iterations after set-off $\dist_H(s,s')\leq k\dist_G(s,s')+O(k^2) W_{max}(G)$.

    \item(Distance improvement) These are iterations after the set-off iteration, and for which $\dist_H(s,s')-\dist_{H'}(s,s')= \Omega(W_{max}(G))$\footnote{As before, $H'$ is the state of $H$ after the completion step.}. To complement this definition, in \Cref{lem:distance-improvement} and \Cref{lem:edge-diff-cent} we show that if $t_i$ edges were added in the $i^{th}$ steps, then $\Omega(t_i)$ \emph{distinct} pairs of centers have their distance improved by $\Omega(W_{max}(G))$.

    \item(Finalized pair) A pair $(s,s')$ is finalized in the $i^{th}$ iteration, if $\dist_H(s,s') \leq k\dist_G(s,s')-k^2W_{max}(G)$, where $H$ is the state of the spanner in that iteration. The key property we use is \Cref{lem:finalized-no-improve}, which says that once a pair is finalized, it can't be part of the $\Omega(t_i)$ improved pairs of \Cref{lem:distance-improvement}.
\end{itemize}

We now explain the interplay between the above analysis and the designed completion. The bound we get on $\dist_H(s,s')$ after the set-off iteration follows since in each greedy step the bottlenecks are added, hence every contiguous union of segments has stretch $\leq k$. Since
$s,s'$ are not on $P_{x,y}$ but are still somewhat close, they inherit the small stretch from the attaching points, up to the additive slack mentioned.

The distance improvement of many pairs $(s,s')$ is quite natural, and follows since the pair $x,y$ chosen does not satisfy the stretch guarantee. A tricky part of the proof is to show $\Omega(t_i)$ of the pairs are distinct, as it may be the case that the number of centers that attach to $P_{x,y}$ is significantly smaller than $t_i$. We show that this cannot happen as otherwise we could have found a shortcut (in other words, the chosen segmentation is not minimal, violating the definition of the completion step). This part is based on adapting a lemma of \cite{AhmedBSKS20} from the additive case to the $(\alpha,\beta)$ set-up.

Finally, showing that finalized pairs are not part of the $\Omega(t_i)$ improving pairs again follows from the minimal segmentation property. \Cref{lem:finalized-no-improve} shows that whenever a finalized pair is one of the improving pairs considered in that iteration, one can make the segmentation smaller, contradicting minimality. We also highlight that here lies a technical difference between the additive spanners and $(\alpha,\beta)$ spanners. When $\alpha=1$ knowing $\dist_H(u,v)\leq \alpha \dist_G(u,v)$ implies $\dist_H(u,v) =  \dist_G(u,v)$ hence the distance in $H$ of $(u,v)$ cannot be improved due to a \emph{tautology}. On the other hand, when $\alpha>1$, pairs $(u,v)$ satisfying  $\dist_H(u,v)\leq \alpha \dist_G(u,v)$ for $H$ corresponding in some iteration can still be improved.

Finally, combining all the above we have that there are $\tilde{O}(|S_R|^2)=\tilde{O}(n^{1+1/k})$ pairs of centers, and each pair improves $O(k^2)$ times. Since one may assume $k=O(\log(n))$ as larger stretch does not improve sparsity, we have that $\tilde{O}(k^2n^{1+1/k})=\tilde{O}(n^{1+1/k})$ pairs of centers of $S_R$ improve in total. Concluding, we get that the number of edges added in the completion step is $ \sum_i t_i = \tilde{O}(n^{1+1/k})$, as needed.

\begin{remark}
    We highlight that one may replace the clustering of \cite{baswana2007simple} with the recent greedy clustering of \cite{PopovaTzalikITCS}, and obtain a deterministic algorithm, as well as a polylogarithmic improvement in the spanner size, though we do not take this route as it complicates notations. By optimizing the methods of this paper, it seems that one can obtain a $(k,k-1)$-spanner of size $O(kn^{1+1/k})$, but it is unclear how to get an optimal $O(n^{1+1/k})$.
\end{remark}

%% file: completion.tex
\section{The Algorithm}

Throughout, let $\G=(\V,\E)$ be a weighted graph, and let $\Hgraph\subseteq \G$ be a subgraph.
We assume that every edge of $\G$ is itself a shortest path between its endpoints; otherwise,
such an edge may be safely removed without affecting the result.

\subsection{Shortest paths and segmentations}

For vertices $x,y\in\V$, let $\Pxy=(x=x_0,x_1,\ldots,x_t=y)$ denote a shortest $x$--$y$ path in $\G$.

\begin{definition}[$\alpha$-segmentation]
An \emph{$\alpha$-segmentation} of the path $\Pxy$ is a sequence
\[
\mathcal{S}=([i_0,i_1],[i_1,i_2],\ldots,[i_{s-1},i_s]),
\qquad
0=i_0<i_1<\cdots<i_s=t,
\]
such that for every segment satisfying $i_k-i_{k-1}\ge 2$,
\[
\distH(x_{i_{k-1}},x_{i_k})
\;\le\;
\alpha\cdot \distG(x_{i_{k-1}},x_{i_k}).
\]
\end{definition}

\begin{definition}[minimal segmentation]
An $\alpha$-segmentation $S$ is minimal if no block of one or more consecutive segments of $S$
can be merged into a single segment so as to produce another $\alpha$-segmentation of the same path.
\end{definition}

Given a segmentation $\mathcal{S}$ of $P_{x,y}$, let $E(\mathcal{S})$ denote the set of all edges corresponding to length $1$ segments in $\mathcal{S}$.

\begin{definition}[Distant edges of a segmentation]
    Given an $\alpha$-segmentation $\mathcal{S}$ of $P_{x,y}$, denote by $E_{\alpha}(\mathcal{S},H) = \left\{\{u,v\}=e\in E(\mathcal{S}) \mid \dist_{H}(u,v)>\alpha w(e) \right\}$.
\end{definition}

\subsection{The $(\alpha,\beta)$-completion.}

We now formally describe the $(\alpha,\beta)$-completion algorithm. Initially, $H \gets H_0$:
\begin{enumerate}
    \item While $H$ is not an $(\alpha,\beta)$-spanner of $G$:
    \item Pick any $x,y$ such that $\dist_{H}(x,y) > \alpha \dist_{G}(x,y) + \beta W_{\max}(G)$.
    \item Find a \emph{minimal segmentation} $\mathcal{S}$ of $P_{x,y}$ (w.r.t. the current $H,G$).
    \item Update $H \gets H \cup E_{\alpha}(\mathcal{S},H)$.
\end{enumerate}

\begin{lemma}\label{lem-single-completion}
    Fix an iteration in the $(\alpha,\beta)$ completion for which $(x,y)$ is the chosen pair violating the desired stretch. Let $\mathcal{S}=([i_0,i_1],\ldots,[i_{s-1},i_s])$ be the chosen minimal segmentation.

    Then in the updated subgraph $H' = H \cup E_{\alpha}(\mathcal{S},H)$, for any $0 \leq \ell < r \leq s$ we have $\dist_{H'}(x_{i_\ell},x_{i_{r}})\leq \alpha \cdot \dist_G(x_{i_\ell},x_{i_{r}})$.
\end{lemma}

\begin{proof}
    We first prove that for any single segment $[i_j,i_{j+1}]$ we have $\dist_{H'}(x_{i_j},x_{i_{j+1}})\leq \alpha \cdot \dist_{G}(x_{i_j},x_{i_{j+1}})$. This implies that \[\dist_{H'}(x_{i_{\ell}},x_{i_r}) \leq \sum_{j=\ell}^{r-1} \dist_{H'}(x_{i_j},x_{i_{j+1}}) \leq \sum_{j=\ell}^{r-1} \alpha \cdot \dist_{G}(x_{i_j},x_{i_{j+1}}) = \alpha \cdot \dist_G(x_{i_\ell},x_{i_r}),\]
    where we use the equality $\sum_{j=\ell}^{r-1} \dist_{G}(x_{i_j},x_{i_{j+1}})= \dist_G(x_{i_\ell},x_{i_r})$ which follows since the vertices $(x_{i_\ell},x_{i_1},\ldots,x_{i_r})$ appear in order along the \emph{shortest} $x$-$y$ path in $G$. 

    To prove that $\dist_{H'}(x_{i_j},x_{i_{j+1}})\leq \alpha \cdot \dist_{G}(x_{i_j},x_{i_{j+1}})$ for all $j$, it suffices to consider segments $[i_j,i_{j+1}]$ for which $\dist_H(x_{i_j},x_{i_{j+1}})>\alpha \dist_{G}(x_{i_j},x_{i_{j+1}})$. By definition of $\alpha$-segmentation, such segments must correspond to single edges on $P_{x,y}$, which by definition are included in $E_{\alpha}(\mathcal{S},H)$. Hence if $\dist_H(x_{i_j},x_{i_{j+1}})>\alpha \dist_{G}(x_{i_j},x_{i_{j+1}})$ then $\dist_{H'}(x_{i_j},x_{i_{j+1}})\leq \dist_{G}(x_{i_j},x_{i_{j+1}})$, and the claim follows as $\alpha \geq 1$.
\end{proof}

\begin{corollary}
    The $(\alpha,\beta)$-completion algorithm terminates.
\end{corollary}

\begin{proof}
    By the above lemma with $\ell=0, r=s$, we have $\dist_{H'}(x,y)\leq \alpha \dist_G(x,y)$, meaning that in all subsequent iterations $(x,y)$ satisfy the stretch requirement.
\end{proof}

\begin{observation}\label{obs:completion-correct}
    The output of the $(\alpha,\beta)$-completion is an $(\alpha,\beta)$-spanner of $G$.
\end{observation}

\begin{proof}
    The algorithm terminates only when $H$ satisfies $\dist_H(x,y) \leq \alpha \dist_G(x,y) + \beta W_{\max}(G)$ for all pairs $x,y \in V$.
\end{proof}

\paragraph*{Computing the $(\alpha,\beta)$-completion.}

The $(\alpha, \beta)$-completion can be computed in polynomial time. By \Cref{lem-single-completion} the loop runs $O(|V|^2)$ iterations. For a fixed iteration it remains to show that a minimal segmentation can be found efficiently. This can be done as follows: 
Start with the trivial segmentation of $P_{x,y}$, and repeatedly merge a block of consecutive segments whenever their union has stretch at most $\alpha$ in $H$. When no further merge is possible, the resulting segmentation is minimal.


\subsection{The $(k,k-1)$ spanner}\label{sec:algorithm}







Run the Baswana-Sen algorithm on $G$, and let $H_0$ be the resulting spanner. Return the $(k,k-1)$ completion of $H_0$.

By \Cref{lem-single-completion}, the output graph is indeed a $(k,k-1)$-spanner, hence it remains to upper bound $|E(H)|$. The size of the initial $H_0$ is $\tilde{O}(n^{1+1/k})$, thus in order to prove \Cref{thm:main} it remains to prove:
\begin{theorem}\label{thm:size-bound-greedy}
    During the $(k,k-1)$-completion, at most $\tilde{O}(n^{1+1/k})$ edges are added.
\end{theorem}

%% file: analysis.tex
\section{The Analysis}\label{sec:analysis}

\subsection{Setup and notation}

We consider a single greedy iteration in which the pair $(x,y)$ violates the stretch condition, and let $\mathcal{S}$ be the minimal segmentation of $P_{x,y}$ used in the completion step.

Let $e_1, e_2, \ldots, e_t$ denote the edges in $E_k(\mathcal{S}, H)$, ordered along $P_{x,y}$ from $x$ to $y$. Each edge $e_i = (u_i, v_i)$ has endpoints whose stretch in $H$ is greater than $k$, i.e., $\dist_H(u_i, v_i) > k \cdot w(e_i)$.

Without loss of generality, we assume $e_1$ is incident to $x$ (so $u_1 = x$) and $e_t$ is incident to $y$ (so $v_t = y$), otherwise set $x_{new}:=u_1$, and $y_{new}:=v_t$. By the property of the segmentation the path from $x$ to $x_{new}$, and $y$ to $y_{new}$ has stretch $\leq k$ so if $x$, $y$ was chosen by the greedy step, the pair $(x_{new},y_{new})$ would also be chosen by the greedy step, using the corresponding subsegmentation $\mathcal{S}'$ of $\mathcal{S}$ which has $E_k(\mathcal{S}',H)=E_k(\mathcal{S},H)$, hence we may analyze the addition of $E_k(\mathcal{S},H)$ w.r.t to the pair $(x_{new},y_{new})$ as the greedy step.
Hence we assume without loss of generality, that $x\in e_1$, and $y\in e_t$.

We focus on the case of $k$ odd, which is cleaner to describe. In the full version of the paper we include an appendix of the minor adaptations in the analysis required for $k$ even. Throughout we set $R:=\frac{k-1}{2}$, and $S:=S_R$ the set of centers. We write $W := W_{\max}(G)$ throughout this section. We assume $k >1$ as the case $k=1$ is trivial. We use the following notations for centers:
\begin{itemize}
    \item $s_x ,s_y $: the $R$-centers of $x,y$ respectively.
    \item For $1<i<t$: for each edge $e_i = (u_i, v_i)$, let $s_{u_i},s_{v_i}$ denote the $R$-centers of $u_i,v_i$, respectively.
\end{itemize}
All mentioned centers exist by using \Cref{thm:baswana-sen}, with $i=R+1$, and observing that a distant edge has stretch strictly bigger than $k$, hence the second alternative of \Cref{thm:baswana-sen} holds for all edges $e_i$. 

We also remark the following observation, which is relevant to the discussion in \Cref{sec:final-remarks}.

\begin{observation}\label{obs:one-missing-edge-stretch}
    Let $x,y$ be vertices, and $\mathcal{S}$ be a minimal segmentation of $P_{x,y}$. Assume $|E_k(\mathcal{S},H)|=1$, then $x,y$ have the desired stretch in $H$.
\end{observation}

\begin{proof}
    Let $e=\{u,v\}$ be the edge in $E_k(\mathcal{S},H)$, and assume without loss of generality, $u$ is closer to $x$ and $v$ is closer to $y$ along $P_{x,y}$. Then: 
    \begin{align*}
\dist_H(x,y)
&\leq \dist_H(x,u) + \dist_H(u,v) + \dist_H(v,y) \\
&\leq k\cdot \dist_G(x,u) + (2k-1)\dist_G(u,v) + k\cdot \dist_G(v,y) \\
&= k\bigl(\dist_G(x,u) + \dist_G(u,v) + \dist_G(v,y)\bigr) + (k-1)\dist_G(u,v) \\
&= k\cdot \dist_G(x,y) + (k-1)w(e).
\end{align*}
\end{proof}

\subsection{Set-off and post set-off Bounds}

\begin{definition}[Set-off]\label{def:set-off}
    A pair of centers  $(s,s')$ is \emph{set off} at the first greedy iteration in which both $s$ and $s'$ serve as centers of some vertices in $\bigcup_{i=1}^{t} V(e_i)$.
\end{definition}

\begin{claim}[Post set-off bounds]\label{clm:post-setoff}
    When a pair of centers $(s,s')$ is set off in an iteration updating $H$ to $H'$, we have:
    \[
    \dist_{H'}(s,s') \leq k \cdot \dist_G(s,s') + (k^2 - 1) W.
    \]
\end{claim}

\begin{proof}
    By definition of set-off, there exist $a, b \in \bigcup_{i=1}^{t} V(e_i)$ such that $s,s'$ are their centers respectively. By the triangle inequality in $H'$:
    \[
    \dist_{H'}(s,s') \leq \dist_{H'}(s, a) + \dist_{H'}(a, b) + \dist_{H'}(b, s').
    \]
    By the Baswana-Sen invariant, $\dist_{H'}(s, a) \leq R \cdot W$ and $\dist_{H'}(b, s') \leq R\cdot W$, giving $\dist_{H'}(s, a) + \dist_{H'}(b, s') \leq (k-1)W$. Since $a, b$ are endpoints of edges in $E_k(\mathcal{S}, H)$, they are segment endpoints in $\mathcal{S}$, so by \Cref{lem-single-completion}, $\dist_{H'}(a, b) \leq k \cdot \dist_G(a, b)$. Hence:
    \[
    \dist_{H'}(s,s') \leq  k \cdot \dist_G(a, b) + (k-1)W .
    \]
    Since $H' \subseteq G$, we have $\dist_G(a, s), \dist_G(b, s') \leq R \cdot W$, so by triangle inequality in $G$: $\dist_G(a, b) \leq \dist_G(a,s)+\dist_G(s,s')+\dist_G(s',b) \leq \dist_G(s,s') + (k-1)W$. Substituting:
    \[
    \dist_{H'}(s,s') \leq  k(\dist_G(s,s') + (k-1)W) + (k-1)W = k \cdot \dist_G(s,s') + (k^2-1)W. \qedhere
    \]
\end{proof}

We also have:
\begin{definition}[Improved pair]\label{def:improved-pair}
    A pair of centers $(s,s')$ is \emph{improved} in a greedy iteration that updates $H$ to $H'$ if it is either \emph{set-off} in that iteration, \emph{or:}
    \[
    \dist_H(s,s') - \dist_{H'}(s,s') \geq \tfrac{1}{2}W.
    \]
\end{definition}

\begin{remark}\label{rmk:choice-of-param}
    The choice of $\tfrac{1}{2}$ in the definition above is arbitrary, and one may choose any smaller constant and the result continues to hold, with a slight increase to the spanner size.
\end{remark}

\begin{claim} \label{clm:center-dist-sum-after-addition}
    Consider a greedy iteration with chosen pair $(x,y)$, a minimal segmentation $\mathcal{S}$, and edges $e_1, \ldots, e_t$ added, updating $H$ to $H'$. For each $1<i<t$ and $s_i \in \{s_{u_i}, s_{v_i}\}$ we have:
    \[\dist_{H'}(s_x,s_i)+\dist_{H'}(s_i,s_y) \leq k \dist_G(x,y) - R(w(e_1)+w(e_t)).\]
\end{claim}

\begin{proof}

Assume without loss of generality $s_i$ is the center of $u_i$, the symmetric case is handled similarly.
Since $e_1$ is added to $H'$ during the completion phase, and $P_{x,y}[v_1,u_i]$ is a contiguous union of segments of $\mathcal{S}$, we can apply \Cref{lem-single-completion} to get $\dist_{H'}(v_1,u_i)\leq k \dist_{G}(v_1,u_i)$, hence by considering the path $s_x \to x \to u_i \to s_i$ and get:
\[\dist_{H'}(s_x,s_i) \leq R w(e_1) +w(e_1) +k \dist_{G}(v_1,u_i) + R w(e_i) \]

We can apply the same analysis to $s_i$ and $s_y$ via the path $s_i \to u_i \to v_i \to u_t \to v_t \to s_y$ (since we assumed $s_i$ is the center of $u_i$), and obtain:

\[\dist_{H'}(s_i,s_y) \leq R w(e_i) +w(e_i) +k \dist_{G}(v_i,u_t) + (R+1) w(e_t) \]

In total, we have in $H'$:

\begin{equation}\label{eq:equation-sum-dist-after}
    \dist_{H'}(s_x,s_i)+\dist_{H'}(s_i,s_y) \leq (R+1) (w(e_1)+w(e_t)) + k \dist_G(v_1,u_t).
\end{equation}

Since $P_{x,y}$ is a shortest $x$ to $y$ path we have $w(e_1)+\dist_G(v_1,u_t)+w(e_t)=\dist_G(x,y)$. Combining this with \Cref{eq:equation-sum-dist-after} and the fact that $R+1=k-R$ implies the claim.

\end{proof}

\begin{lemma}[Distance improvement]\label{lem:distance-improvement}
    Consider a greedy iteration with chosen pair $(x,y)$, a minimal segmentation $\mathcal{S}$, and edges $e_1, \ldots, e_t$ added, updating $H$ to $H'$. For each $1<i<t$ we have for $s_i \in \{s_{u_i}, s_{v_i}\}$:
    \begin{enumerate}
        \item $\dist_H(s_x, s_i) - \dist_{H'}(s_x, s_i) \geq  \tfrac{1}{2}W$, or
        \item $\dist_H(s_{v_i}, s_y) - \dist_{H'}(s_{v_i}, s_y) \geq  \tfrac{1}{2}W$.
    \end{enumerate}
\end{lemma}

\begin{proof}
    Assume towards contradiction both items do not hold. Then by \Cref{clm:center-dist-sum-after-addition} we have:
    \[\dist_H(s_x,s_y)\leq k\dist_G(x,y) - R(w(e_1)+w(e_t))+W.\]

    By the Baswana-Sen invariant we have $\dist_H(x,s_x)\leq Rw(e_1)$, and $\dist_H(s_y,y)\leq Rw(e_t)$, hence by the triangle inequality:
    \[\dist_H(x,y)\leq k\dist_G(x,y) +W.\]
    Which contradicts the assumption that $(x,y)$ were chosen during the completion phase.
\end{proof}
The following observation is proved in exactly the same way as \Cref{clm:center-dist-sum-after-addition} and \Cref{lem:distance-improvement}, hence omitted.
\begin{observation}\label{obs:farth-cent-pair-improv}
    The pair of centers $(s_x,s_y)$, satisfies:
    \begin{itemize}
        \item $\dist_{H'}(s_x,s_y) \leq k \dist_G(x,y) - R(w(e_1)+w(e_t))$.
        \item $\dist_H(s_x, s_y) - \dist_{H'}(s_x, s_y) \geq  \tfrac{1}{2}W.$
    \end{itemize}
    
\end{observation}

We now define the finalized pairs.

\begin{definition}[Finalized pair]\label{def:finalized}
    A pair of centers $(s,s') $ is \emph{finalized} in $H$ if
    \[
    \dist_{H}(s,s') \leq k \cdot \dist_G(s,s') - k^2 W.
    \]
\end{definition}

\begin{lemma}[Finalized pairs do not participate in completion]\label{lem:finalized-no-improve}
    Consider a greedy completion iteration with chosen pair $(x,y)$ and edges $e_1, \ldots, e_t$. Let $a,b$ be vertices in $\bigcup_{i=1}^{t} V(e_i)$, and let $s,s'$ be their centers respectively. Then if $(s,s')$ is finalized then $a,b$ are the two endpoints of a segment of $\mathcal{S}$ that has stretch $\leq k$.
\end{lemma}

\begin{proof}
    Let $a,b,s,s'$ be as stated. We show that if $(s, s')$ is finalized in $H$, then $\dist_H(a,b)\leq k \cdot \dist_G(a,b)$. This implies the desired result as $\mathcal{S}$ is a \emph{minimal} segmentation, hence the only possibility is that $a,b$ are the endpoints of a segment in $\mathcal{S}$ whose stretch is at most $k$.
    
    By triangle inequality for the $a\to s \to s' \to b$ in $H$, and the Baswana-Sen invariant we have $\dist_H(a,b) \leq (k-1)W+\dist_H(s,s')$. 
    Also, since $(s,s')$ is assumed to be finalized $\dist_H(s,s')\leq k\cdot\dist_G(s,s')-k^2W$.

    Now by the triangle inequality for $s\to a \to b \to s'$ in $G$ we have $\dist_G(s,s')\leq \dist_G(a,b) + (k-1)W$.
    Hence combining it all together we get:
    
    \[\dist_H(a,b) \leq (k-1)W+ k(\dist_G(s,s')-k^2W) < k\cdot \dist_G(a,b) .\]
\end{proof}

We proceed by noting the following lemma, a small variation of Lemma 1 in \cite{AhmedBSKS20}.

\begin{lemma}\label{lem:edge-diff-cent}
    Let $e_1,\ldots,e_t$ be the edges added during the greedy completion phase.
    Then there exist a sub collection of $t^*\geq \min\{\frac{t}{4}-2,1\}$ edges $e^*_1 = (a_1,b_1),\ldots,e^*_{t^*}=(a_{t^*},b_{t^*})$, such that either the centers of all the vertices $\{a_i\}_{i=1}^{t^*}$ are distinct, or the centers of $\{b_i\}_{i=1}^{t^*}$ are distinct.
\end{lemma}

\begin{proof}
    If $t<12$ there is nothing to prove (any collection of a single edge will do), so assume $t\geq 10$. For $1<i<t$, say that the $i^{th}$ edge, $e_i$ is pre-heavy if $w(e_{i-1})\leq w(e_i)$, i.e. it is heavier than the edge preceding it (among the added edges). Otherwise, it is pre-light (ties are broken arbitrarily). As every internal edge is either pre-heavy or pre-light, $\geq \frac{t-2}{2}$ edges are pre-heavy or pre-light, assume without loss of generality that they are pre-light.
    
    We now claim that if $e'_1 = (a'_1,b'_1),\ldots e'_{t'}=(a'_{t'},b'_{t'})$\footnote{We assume that the edges are ordered according to their order on the path $P_{x,y}$. That is $e'_i$ is closer to $x$ then $e'_j$ for $i<j$, and that $a'_i$ is closer to $x$ then $b'_i$} are pre-light, then every center $s$ can be the center of at most $2$ vertices in $\{a'_i\}_{i}$.
    
    To prove the claim, assume we have $i<j<m$ such that $s$ is the center of $a'_i,a'_j,a'_m$. From the initialization step we have: $\dist_H(a'_i,a'_m)\leq \dist_H(a'_i,s)+\dist_H(s,a'_m) = R (w(e'_i)+w(e'_m))$.

    Let $e'$ be the edge preceding $e'_m$ on $P_{x,y}$. Since $e'_m$ is pre-light we have $w(e')\geq w(e'_m)$. Also since $i<j<m$, it must be that $e'$ and $e'_i$ are different edges, which both lie on $P_{x,y}[a'_i,a'_m]$. We obtain $\dist_G(a'_i,a'_m)\geq (w(e'_i)+w(e'))$, and we get:  
    \[\dist_H(a'_i,a'_m)\leq R(w(e'_i)+w(e_m'))\leq  R(w(e'_i)+w(e'))\leq R \dist_G(a'_i,a'_m).\]
    Since $R\leq k$ we get a contradiction since all the segments between $a'_i,a'_m$ could be merged to one segment with stretch $\leq k$, contradicting the minimality of the segmentation.

    Finally, we got a collection of $\geq \frac{t-2}{2}$ edges such that every center $s$ is the center left endpoint of at most two of them, hence we can find a subcollection of at least $\frac{t-2}{4}$ edges whose left centers are all distinct\footnote{Take any maximal subcollection $E'$ of edges whose left centers are all different, since for every edge $e\not\in E'$ there must be another edge $e^* \in E'$ with the left center of $e$ as the left center of $e^*$, and this mapping is injective (as each center is the left center of $\leq 2$ edges), there are more edges inside the maximal subcollection than outside it - the claim follows. }.

    When most internal edges are pre-heavy, we consider them as post-light. In such a case we find a collection of $\Omega(t)$ edges with distinct right centers.
\end{proof}

\begin{corollary}\label{cor:improve-not-final-each-round}
    In a round with $t$ edges added, $\Omega(t)$ pairs of centers that are not finalized are improving.
\end{corollary}

\begin{proof}
    If $t<100$ then by \Cref{obs:farth-cent-pair-improv} we have a single pair of centers improving. Otherwise by \Cref{lem:edge-diff-cent} we have $\geq t/4-2$ edges with distinct (without loss of generality) left centers. By removing at most four edges from that collection, those whose left centers are $s_x$ or $s_y$, we have $\Omega(t)$ edges with distinct left centers, that are not $s_x$ or $s_y$. By \Cref{lem:distance-improvement}, all these centers give an improved pair with either $s_x$ or $s_y$, and the resulting center pairs are all distinct. The considered pairs of centers in \Cref{lem:distance-improvement} all attach to either $x$ or $y$ and another vertex on $P_{x,y}$, hence they attach to $P_{x,y}$ to two points that don't form a good stretch segment, hence by \Cref{lem:finalized-no-improve} all these pairs of centers are not finalized.
\end{proof}
    
We now complete the proof of \Cref{thm:size-bound-greedy}:

\begin{proof}
    Let $t_i$ be the number of edges added in the $i^{th}$ round of the completion process, and let $m_i$ be the number of center pairs whose distance improved in the $i^{th}$ round according to \Cref{lem:distance-improvement}, which are not finalized. By \Cref{cor:improve-not-final-each-round} we have $m_i = \Omega(t_i)$. As each pair of centers can improve $O(k^2)$ times before it becomes finalized and there are $\tilde{O}(n^{1+1/k})$ center pairs we have 
    \[\Omega \left(\sum t_i\right)=\sum m_i = \tilde{O}(k^2 n^{1+1/k})\]
    Since $k=O(\log(n))$ we get that the number of edges added during the completion phase, $\sum t_i$, is $\tilde{O}(n^{1+1/k})$, as required.
\end{proof}

%% file: Conclusions.tex
\section{Final Remarks and Open Problems}\label{sec:final-remarks}

A main open question left from this work is whether unweighted $(\alpha,\beta)$-spanners can be extended to the weighted setup? There are many $(\alpha,\beta)$ for which $(\alpha,\beta)$-spanners exist for unweighted graphs, though almost all constructions do not generalize; e.g., it is open whether the construction of \cite{Ben-LevyP20}, as well as the recent constructions of \cite{ChechikL26} which have $o(k)=\alpha = \Omega(1)$ can be extended to the weighted setting.

\begin{question}
    Is it the case that if there is an unweighted $(\alpha,\beta)$-spanner for some $(\alpha,\beta)$, then there is a weighted $(\alpha, \beta)$-spanner with a similar size guarantee.
\end{question}

Recall that the lightness of a spanner is the ratio of its total weight divided by the weight of a minimum spanning tree. We ask:

\begin{question}
    For which $(\alpha,\beta)$ can we provide an $(\alpha,\beta)$-spanner with relatively small lightness?
\end{question}

Finally we propose the following more structural direction:
\paragraph*{Where did the $k-1$ go?} The keen reader may have noticed that in \Cref{sec:algorithm} the fact that $\dist_H(x,y) > k\cdot \dist_G(x,y) + (k-1)W$ was not used. It actually appears only in \Cref{obs:one-missing-edge-stretch}, hence under the assumption that $(x,y)$ have a minimal segmentation with \emph{more than} one missing edge we can run the completion algorithm under the same size bound and give the guarantee $\dist_H(x,y) \leq k\dist_G(x,y) + W_{max}(G)$. As explained in \Cref{rmk:choice-of-param} the choice of the constant $1$ is arbitrary. In fact, one can achieve the following: For any $\varepsilon>0$, and any weighted graph $G$ there exist a subgraph $H$ of size $\Tilde{O}(n^{1+1/k}/\varepsilon)$ such that for any $x,y$, and any shortest path $P_{x,y}$ between them in $G$ one has:
\begin{itemize}
    \item($(k,k-1)$-local stretch) $\dist_H(x,y) \leq k \cdot \dist_G(x,y) + (k-1) W_{max} (P_{x,y})$, where $P_{x,y}$ is an $x$ to $y$ shortest path. Or,

    \item (stretch $k$ and additive $+\varepsilon W_{max}(G)$) Namely, $\dist_H(x,y) \leq k \cdot \dist_G(x,y) + \varepsilon W_{max} (G)$.
\end{itemize}

The proof follows from the analysis in \Cref{sec:analysis}. Pick a minimal segmentation $\mathcal{S}$ of $P_{x,y}$, and distinguish paths according to $|E_k(\mathcal{S},H)|$ . If $|E_k(\mathcal{S},H)|=0$ the first alternative above clearly holds. If $|E_k(\mathcal{S},H)|=1$, then by \Cref{obs:one-missing-edge-stretch} we have $\dist_H(x,y) \leq k \cdot \dist_G(x,y)+(k-1)w(e)$ where $e$ is the missing edge, hence the first guarantee holds. Finally if $|E_k(\mathcal{S},H)|\geq 2$, then we check if the second guarantee holds, and if so we do not do anything, if not, we add $E_k(\mathcal{S},H)$. By changing the definition of improved pairs,  \Cref{def:improved-pair} to have distance difference between $H$, $H'$ bigger than $\frac{\varepsilon}{2} W$. All parts of the proof stay exactly the same, except for the final proof of  \Cref{thm:size-bound-greedy}, which needs to be augmented as every pair of centers can improve $O(k^2/\varepsilon)$ times, yielding an $1/\varepsilon$ larger spanner.

\begin{question}
    Can one provide a better structural guarantee than being just an $(\alpha,\beta)$ spanner for weighted graph? E.g. is it possible to ensure that for every $x,y$ and any shortest path $P_{x,y}$, can be segmented into $O(\beta)$ edges from $G$ and paths of stretch $\alpha$ plus a relatively negligible error?
\end{question}

It is plausible that resolving the question above can provide a path to allow comparison for the guarantees of $(\alpha,\beta)$-spanners, for different $\alpha$'s and $\beta$'s, similarly to the unweighted setup.
